\documentclass[draftcls,11pt,onecolumn,romanappendices]{IEEEtran}

\usepackage{times}
\usepackage{amsmath}  
\usepackage{amssymb} 
\usepackage{mathrsfs}
\usepackage{cite}
\usepackage{comment} 
\usepackage{upref}
\usepackage{amsfonts}
\usepackage{verbatim}
\usepackage[dvipsnames,usenames]{color}
\usepackage{enumerate}
\usepackage{graphicx}
\usepackage{latexsym}
\usepackage{bm}
\usepackage{multirow}
\usepackage{dsfont}
\usepackage{amsthm,xpatch}
\usepackage{mathtools}
\usepackage{bbm}
\usepackage{latexsym}
\usepackage{accents}
\usepackage{psfrag}
\usepackage{soul}
\usepackage{color}
\usepackage{times,color}
\usepackage[usenames,dvipsnames,svgnames,table]{xcolor}
\usepackage{pgf}
\usepackage{tikz}
\usetikzlibrary{arrows, automata, positioning, calc}
\usepackage{cancel} 
\usepackage{tikz, tikz-3dplot, amssymb}
\usepackage{graphicx}
\usepackage{enumerate}
\usepackage{courier}
\usepackage{times}
\usepackage{amsmath,amssymb}
\usepackage{color}
\usepackage{algorithm}
\usepackage{algorithmic}

\newtheorem{lemma}{Lemma}
\newtheorem{Corollary}{Corollary}
\newtheorem{theorem}{Theorem}

\newtheorem{proposition}{Proposition}
\newtheorem{definition}{Definition}
\theoremstyle{definition}
\newtheorem{example}{Example}

\newcommand{\conv}{\operatorname{conv}}

\newcommand{\PG}{\operatorname{PG}}

\newcommand{\vek}[1]{\mathbf{#1}}
\newcommand{\mat}[1]{\mathbf{#1}}

\newcommand{\removelatexerror}{\let\@latex@error\@gobble}
\newcommand{\cS}{\mathcal{S}}
\newcommand{\cC}{\mathcal{C}}
\newcommand{\cE}{\mathcal{E}}
\newcommand{\cH}{\mathcal{H}}
\newcommand{\cY}{\mathcal{Y}}
\newcommand{\cG}{\mathcal{G}}
\newcommand{\cR}{\mathcal{R}}
\newcommand{\cV}{\mathcal{V}}
\newcommand{\cX}{\mathcal{X}}
\newcommand{\bG}{\mathbf{G}}
\newcommand{\bg}{\mathbf{g}}
\newcommand{\be}{\mathbf{e}}
\newcommand{\bn}{\mathbf{n}}

\newcommand{\bv}{\mathbf{v}}
\newcommand{\bbF}{\mathbb{F}}
\newcommand{\bbN}{\mathbb{N}}
\newcommand{\bbZ}{\mathbb{Z}}
\newcommand{\bbRp}{\mathbb{R}_{\ge 0}}
\newcommand{\bbR}{\mathbb{R}}
\newcommand{\bx}{\mathbf{x}}
\newcommand{\by}{\mathbf{y}}

\newcommand{\blambda}{\boldsymbol{\lambda}}
\newcommand{\simplex}{\operatorname{Simpl}}

\renewcommand{\mathsf}[1]{#1}

\makeatother
\makeatletter

\IEEEoverridecommandlockouts

\begin{document}

\pagestyle{plain}

\title{\LARGE \textbf{\Large Efficient Storage Schemes for Desired Service Rate Regions
}}

\author{\normalsize {$^\ast$}Fatemeh Kazemi, {$^\dagger$}Sascha Kurz, {$^\ddagger$}Emina Soljanin, and {$^\ast$}Alex Sprintson \\{\small {$^\ast$}Dept. of ECE, Texas A\&M University, USA (E-mail: \{fatemeh.kazemi, spalex\}@tamu.edu)}\\{\small {$^\dagger$} Dept. of Mathematics, University of Bayreuth, Germany (E-mail: sascha.kurz@uni-bayreuth.de)}\\ {\small {$^\ddagger$}Dept. of ECE, Rutgers University, USA (E-mail: emina.soljanin@rutgers.edu)}\thanks{Part of this research is based upon work supported by the National Science Foundation under Grants No. CIF-1717314, as well as work while some authors were in residence at the Schloss Dagstuhl Research Institute during the Algebraic Coding Theory for Networks, Storage, and Security Seminar in 2018. The authors thank Yiwei Zhang for helpful discussion.}}

\maketitle 

\thispagestyle{plain}

\begin{abstract} 
A major concern in cloud/edge storage systems is serving a large number of users simultaneously. The service rate region is introduced recently as an important performance metric for coded distributed systems, which is defined as the set of all data access requests that can be simultaneously handled by the system. This paper studies the problem of designing a coded distributed storage system storing $k$ files where a desired service rate region $\cR$ of the system is given and the goal is 1) to determine the minimum number of storage nodes ${n(\cR)}$ (or a lower bound on ${n(\mathcal{R})}$) for serving all demand vectors inside the set $\cR$ and 2) to design the most storage-efficient redundancy scheme with the service rate region covering $\cR$. Towards this goal, we propose three general lower bounds for $n(\cR)$. Also, for ${k=2}$, we characterize ${n(\cR)}$, i.e., we show that the proposed lower bounds are tight via designing a novel storage-efficient redundancy scheme with ${n(\cR)}$ storage nodes and the service rate region covering $\cR$.    
\end{abstract}

\section{Introduction}

\textbf{Motivation:} 
The past two decades have seen an explosive growth in the amount of data stored in the cloud data centers which was accompanied by a rapid increase in the volume of users accessing it. To handle these surging demands in a fast, reliable and efficient manner, chunks of a data object are stored redundantly over multiple storage nodes through either replication or erasure coding. Although replication has been typically preferred due its simplicity, it can be expensive in terms of storage. Erasure codes have been shown to be effective in achieving various goals such as providing reliability against node failures (see e.g., \cite{dimakis2010network,dimakis2011survey,rashmi2009explicit,rashmi2011optimal}), ensuring availability of stored content during high demand (see e.g., \cite{rabinovich2002web,maddah2016coding,shanmugam2013femtocaching,tran2017mobee}), enabling the recovery of a data object from multiple disjoint groups of nodes (see e.g., \cite{rawat2016locality,tamo2014bounds,papailiopoulos2014locally}), and providing fast content download (see e.g., \cite{joshi2012coding,liang2014fast,latency:JoshiSW15,latency:JoshiSW17,shah2014mds,gardner2015reducing,shah2015redundant,kadhe2015analyzing,aktas2019analyzing}).

Serving a large number of users simultaneously is a major concern in cloud storage systems and so is considered as one of the most significant considerations in the design of coded distributed systems. The service rate region has been recently recognized as an important performance metric for coded distributed systems, which is defined as the set of all data access requests that can be simultaneously served by the system \cite{aktas2020service,aktacs2017service,anderson2018service,ServiceCombinatorial:KazemiKSS20,ServiceGeometric:KazemiKS20, peng2018distributed,noori2016storage}. It has been shown that erasure coding of data objects can increase the overall volume of the service rate region through handling skews in the request rates more flexibly  \cite{aktas2020service,anderson2018service,aktacs2017service}.

The service rate problem considers a distributed storage system in which $k$ files $f_1,\dots,f_k$ are stored across $n$ servers using a linear ${[n,k]_q}$ code. The requests to download file $f_i$ arrive at rate $\lambda_i$, and the service rate of each server is $\mu$. A goal of the service rate problem is to determine the service rate region of this  system which is the set of all request rates $\boldsymbol{\lambda}=(\lambda_1,\dots,\lambda_k)$ that can be served by this system.

\vspace{0.2cm}
\textbf{Previous Work:} 
All the existing studies on the service rate problem focus on characterizing the service rate region of a given coded storage scheme and finding the optimal request allocation, that is, the optimal policies for splitting incoming requests across the nodes to maximize the service rate region (see \cite{aktas2020service}). In~\cite{aktacs2017service}, the service rate region was characterized for MDS codes when $n \geq 2k$, binary simplex codes and systems with arbitrary $n$ when ${k=2}$ . The service rate region of the systems with arbitrary $n$ when ${k=3}$ was determined in~\cite{anderson2018service}. A connection between the service rate problem and the fractional matching problem is established in \cite{ServiceCombinatorial:KazemiKSS20}. Also, it has been shown that the service rate problem can be viewed as a generalization of the multiset primitive batch codes problem. In \cite{ServiceGeometric:KazemiKS20}, we characterized the service rate regions of the binary first order Reed-Muller codes and binary simplex codes using a novel geometric technique. Also, we showed that given the service rate region of a code, a lower bound on the minimum distance of the code can be derived.

\vspace{0.2cm}
\textbf{Main Contributions:}
In this paper, we consider a practical setting of designing a coded distributed storage system where we are asked to store $k$ files redundantly across some number of storage nodes in the system. Also, we are given a bounded subset $\mathcal{R} \subset \mathbb{R}^k_{\geq 0}$ as a desired service rate region for this distributed storage system. Our goal is: 1) to find the minimum number of storage nodes $n(\mathcal{R})$ (or a lower bound on $n(\mathcal{R})$) required for serving all demand vectors $\boldsymbol{\lambda}$ inside the desired service rate region $\mathcal{R}$, and 2) to design the most storage-efficient redundancy scheme with the service rate region covering the set $\mathcal{R}$. In fact, in this paper, unlike the existing work, we focus on designing the underlying erasure code for covering a given service rate region with minimum storage. Towards this goal, we propose three different general lower bounds for $n(\mathcal{R})$. Also, we show that for $k=2$, these bounds are tight and we design an efficient storage scheme that achieves the desired service rate region while minimizing the storage. \textit{All the proofs can be found in the Appendix}.

\section{Problem Setup and Formulation}\label{sec:Problem Statement}

\subsection{Basic Notation}
Throughout this paper, we denote vectors by bold-face lower-case letters and matrices by bold-face capital letters. Let ${\bbZ_{\geq 0}}$ and $\bbN$, respectively, denote the set of non-negative integers, and the set of positive integers. For ${k\in \bbN}$, let $\vek{0}_k$ and $\vek{1}_k$, respectively, denote the all-zero and all-one column vectors of length $k$. Let $\vek{e}_i$ be a unit vector of length $k$, having a one at position $i$ and zeros elsewhere. For any $i \in \bbN$, we define $[i]\triangleq\{1,\dots,i\}$. Let $\bbF_q$ be the finite field of order $q$, and $\bbF_q^n$ be the $n$-dimensional vector space over $\bbF_q$. Let $[n,k]_q$ denote a q-ary linear code $\cC$ of length $n$ and dimension $k$. We denote the cardinality of a set or multiset $\cS$ by $\#\cS$. Let $\left\langle \cS \right\rangle$ and $\conv(\cS)$, respectively, denote the span and the convex hull of the set $\cS$ of vectors. For two vectors $\bx=(x_1,\dots,x_k)$ and $\by=(y_1,\dots,y_k)$, let $\bx\le \by$ define $x_i\le y_i$ for all $i\in [k]$.

\subsection{Coded Storage System} 

Consider a coded storage system wherein $k$ files $f_1,\dots,f_k$ are stored redundantly across $n$ servers using a linear code of length $n$ and dimension $k$ over $\mathbb{F}_q$ with generator matrix $\bG$. Suppose all files are of the same size, and all servers have a storage capacity of one file. A set $Y$ is a recovery set for file $f_i$ if the unit vector $\be_i$ can be recovered through a linear combination of the columns of $\bG$ indexed by the set $Y$, i.e., if there exist coefficients $\alpha_j$'s $\in \mathbb{F}_q$ such that ${\sum_{j \in Y}\alpha_j \bg_j=\mathbf{e}_i}$ where $\bg_j$ denotes the $j$th column of $\bG$. For each file, w.l.o.g. we consider reduced recovery sets defined as the recovery sets that are not a proper superset of any other recovery sets for that file. In other words, the reduced recovery sets are obtained by considering non-zero coefficients $\alpha_j$'s and linearly independent columns $\bg_j$'s. Let ${\cY_{i}=\{Y_{i,1},\dots,Y_{i,t_i}\}}$ denote the $t_i$ recovery sets for file ${f_i}$. 

We assume that the service rate of each server is $\mu$, i.e., each server can resolve the received requests at the average rate $\mu$. We further assume that the requests to download file $f_i$ arrive at rate $\lambda_i$, $i\in [k]$. The request arrival rates for the $k$ files are denoted by the demand vector $\blambda=(\lambda_1,\dots,\lambda_k)$. We consider the class of scheduling strategies that assign a fraction of requests for a file to each of its recovery sets. Let $\lambda_{i,j}$ be the portion of requests for file $f_i$ that is assigned to the recovery set $Y_{i,j}$, $j\in [t_i]$.

\subsection{Service Rate Region}
The demand vector $\blambda$ can be served by a coded distributed storage system with generator matrix $\bG\in\bbF_q^{k\times n}$ and service rate ${\mu}$ iff there exists a set ${\{\lambda_{i,j}:i\in [k], j\in [t_i]\}}$, referred to as a valid allocation, that satisfies the follwing constraints:

\begin{subequations}\label{eq:1}
\begin{align}\label{condition_demand}
&\sum_{j=1}^{t_i} \lambda_{i,j}=\lambda_i, ~~~~~~~~~~~ \text{for all} ~~~ i\in [k], \\
&\sum_{i=1}^{k}\sum_{\substack{{j\in [t_i]} \\  \label{condition_capacity} {\hspace{0.15cm}\ell \in Y_{i,j}}}} \lambda_{i,j} \leq \mu, ~~~~ \text{for all} ~~~ \ell\in [n],\\
&\hspace{0.1cm}\lambda_{i,j} \in \mathbb{R}_{\geq 0},~~~~~~~~~~~~~~ \text{for all} ~~~ i\in [k],~j\in [t_i].\label{condition_pos}
\end{align}
\end{subequations}

The constraints~\eqref{condition_demand} guarantee that the demands for all files are served, and the constraints~\eqref{condition_capacity} ensure that the total rates assigned to each server is not more than its service rate. 

The \emph{service rate region} of an erasure coded storage system with the generator matrix $\bG$ and service rate $\mu$, denoted by $\cS(\bG,{\mu}) \subseteq \bbR^k_{\geq 0}$, is defined as the set of all demand vectors $\blambda$ that can be served by the system. In what follows, w.l.o.g. we assume that $\mu=1$ and abbreviate $\cS(\bG,1)$ as $\cS(\bG)$.

Note that there are several generator matrices that span the same linear code, i.e., whenever the row span of two matrices $\bG$ and $\bG'$ coincides, they span the same code. However, the service rate regions of generator matrices $\bG$ and $\bG'$ of the same linear code might not be the same, i.e., $\cS(\bG)\neq \cS(\bG')$.

\subsection{Geometric description of Linear Codes}
Here, we briefly review some preliminaries regarding the notions of projective space, multiset, and projective multisets induced by linear codes that we will use in Sec.\ref{sec:geomservice}. For more details, see~\cite{tsfasman1995geometric,dodunekov1998codes,beutelspacher1998projective}.

\begin{definition}
For a vector space $\cV$ of dimension $v$ over $\bbF_q$, the \ul{projective space} of $\cV$, denoted as ${\operatorname{PG}(\cV)}$, is the set of equivalence classes of $\cV\setminus\{\vek{0}_v\}$ under the equivalence relation $\sim$ defined as ${x} \sim {y}$ if there is a non-zero element ${\alpha \in \mathbb{F}_q}$ such that ${x}=\alpha{y}$.
\end{definition}

We remark that the $1$-dimensional subspaces of ${\cV}$ are the points of the projective space ${\operatorname{PG}(\cV)}$. The $2$-dimensional subspaces of $\cV$ are the lines of ${\operatorname{PG}(\cV)}$ and the ${v-1}$ dimensional subspaces of $\cV$ are called the hyperplanes of $\operatorname{PG}(\cV)$. 

For a vector space ${\cV}$ of (\emph{geometric}) dimension $v$ over ${\bbF_q}$, the projective space ${\operatorname{PG}(\cV)}$ is also denoted by $\operatorname{PG}(v-1,q)$, referred to as the projective space of (\emph{algebraic}) dimension ${v-1}$ over ${\bbF_q}$. This notion makes sense since up to isomorphism, the ${\PG(\cV)}$ only depends on the order $q$ of the base field and the dimension $v$ of the vector space $\cV$. Thus, $\operatorname{PG}(v-1,q)$ can be defined as the set of $v$-tuples of elements of ${\bbF_q}$, not all zero, under the equivalence relation given by $(x_1,\cdots,x_v)\sim(\alpha x_1,\cdots,\alpha x_v)$, $\alpha \neq 0$, $\alpha \in \bbF_q$. The definition implies that if $(x_1,\cdots,x_{v})$ is a point in $\operatorname{PG}(v-1,q)$, its scalar multiple (by any non-zero scalar $\alpha \in \bbF_q$) $(\alpha x_1,\cdots,\alpha x_{v})$ is the same point in $\operatorname{PG}(v-1,q)$.

A \textit{multiset}, unlike a set, allows for multiple instances for each of its elements. A multiset $\cS$ on a base set $\cX$ is defined with its characteristic function, denoted as ${\chi_{\cS} : \cX \to \bbN}$, mapping $x \in \cX$ to the multiplicity of $x$ in $\cS$. The cardinality of the multiset $\cS$ is computed as $\#\cS = \sum_{x\in X} \chi_{\cS}(x)$. The multiset $\cS$ is also called \emph{$\#\cS$-multiset}. As a simple example, consider the multiset $\cS=\{a,a,b,b,b,c\}$ on the base set $\cX=\{a,b,c\}$ that is identified with $\chi_{\cS}(a)=2$, $\chi_{\cS}(b)=3$ and $\chi_{\cS}(c)=1$.

Let $\mat{G}$ be the generator matrix of an $[n,k]_q$ code $\cC$ that is a $k$-dimensional subspace of the $n$-dimensional vector space $\bbF_q^n$. Let $\bg_i$, ${i\in [n]}$ be the $i$th column of $\mat{G}$. Then, each $\bg_i$ is a point in the projective space ${\operatorname{PG}(k-1,q)}$, and $\cG:=\{{g}_1,{g}_2,\dots,{g}_n\}$ is an $n$-multiset of points in ${\operatorname{PG}(k-1,q)}$ where each point is counted with the appropriate multiplicity. In general, $\cG$ is called the $n$-multiset induced by $\cC$.

\begin{proposition}\label{prop:corres}
There exists a one-to-one correspondence between the equivalence classes of full-length q-ary linear codes and the projective equivalence classes of multisets in finite projective spaces. 
\end{proposition}

An $[n,k]_q$ code can be described by a generator matrix $\bG$ or as discussed by an n-multiset $\cG$ of points in $\PG(k-1,q)$. In what follows, for the ease of notation, we restrict ourselves to the binary field. We associate the points of $\PG(k-1,2)$ with the non-zero vectors in $\bbF_2^k \setminus \{\vek{0}_k\}$, then we interpret each such vector as the binary expansion of the corresponding integer $i\in [\ell]$ where $\ell:=2^k-1$. We denote by $\bv_i$ the vector corresponding to the integer $i\in[\ell]$. As two examples, in $\bbF_2^3 \setminus \{\vek{0}_3\}$, the vectors $\bv_3=(0,1,1)$ and $\bv_4=(1,0,0)$ are corresponding to the integers $3$ and $4$, respectively. In order to uniquely characterize a multiset of points $\cG$ in $\PG(k-1,2)$, we use multiplicities $n_i\in \bbZ_{\geq 0}$, $i\in[\ell]$, counting the number of occurrences of the vector $\bv_i \in \bbF_2^k\setminus \{\vek{0}_k\}$, $i\in[\ell]$, in the generator matrix $\bG$. Thus, we have $\sum_{i\in[\ell]} n_i=n$. Also, due to the correspondence between a generator matrix $\bG$ and a multiset of points $\cG$ (based on the Proposition~\ref{prop:corres}), we can write $\cS(\cG)$ instead of $\cS(\bG)$ for the service rate region and we will directly define $\cS(\cG)$ shortly.

\subsection{Geometric Interpretation of the Service Rate Region}\label{sec:geomservice}

A recovery set for file $f_i$, $i\in [k]$, is a subset $Y\subseteq [\ell]$ such that the span of the set $ \left\{\bv_j\mid j\in Y\right\}$ contains the unit vector $\be_i$. A recovery set $Y$ for $f_i$ is called reduced if there does not exist a proper subset $Y'\subsetneq Y$ with $\be_i\in \left\langle \left\{\bv_j\mid j\in Y'\right\}\right\rangle$. For $q=2$ and a reduced recovery set $Y$, there is no need to specify the index $i$ of the file that is recovered since $\sum_{j\in Y} \bv_j=\be_i$. However, this is not necessarily true for $q > 2$. As an example, in $\bbF_3$ the set $\left\{\be_1+\be_2,\be_1+2\be_2\right\}$ spans a $2$-dimensional subspace containing both $\be_1$ and $\be_2$, while none of these two unit vectors are contained in the span of a proper subset. Since we assume that $q=2$, we will mostly speak just of a recovery set without specifying the index $i$ of the file that it recovers. By $\cY_i$, we denote the set of all reduced recovery sets for file $f_i$, where $i\in[k]$. As an example, for $k=3$ we have $\cY_2=\{\{2\},\{4,6\},\{1,3\},\{5,7\},\{1,4,7\}\}$. Note that the maximum cardinality of a reduced recovery set is $k$, which can indeed be attained. 

Let $\alpha_{i,Y}$ be the portion of request rates for file $f_i$ assigned to the recovery set $Y\in \cY_i$. Given a multiset of points $\cG$ in $\PG(k-1,2)$, described by the multiplicities $n_j$, $j\in[\ell]$, the service rate region $\cS(\cG)$ is the set of all vectors $\blambda\in\bbRp^k$ for which there exist $\alpha_{i,Y}$'s, satisfying the following constraints: 
\begin{subequations}\label{eq:11}
\begin{align}\label{condition_demand_2}
&\sum_{Y\in\cY_i} \alpha_{i,Y}=\lambda_i, ~~~~~~~~~~~ \text{for all} ~~~ i\in [k], \\
&\sum_{i=1}^{k}\sum_{\substack{{Y\in\cY_i} \\  \label{condition_capacity_2} {\hspace{0.05cm} j\in Y}}} \alpha_{i,Y} \leq n_j, ~~~~~~ \text{for all} ~~~ j\in [\ell],\\
&\hspace{0.1cm}\alpha_{i,Y} \in \mathbb{R}_{\geq 0},~~~~~~~~~~~~~~~ \text{for all} ~~~ i\in [k],~Y\in\cY_i.\label{condition_pos_2}
\end{align}
\end{subequations}

Recall that the constraints~\eqref{condition_demand_2} guarantee that the demands for all files are served, and constraints~\eqref{condition_capacity_2} certify that no node receives requests at a rate in excess of its service rate. 

As noted, for $q=2$, each reduced recovery set uniquely characterizes the file it recovers, that is, $\cY_i$'s where $i \in [k]$ are pairwise disjoint and form a partition of $\cY:=\cup_{i\in[k]} \cY_i$. With this we can simplify the above characterization, i.e., the service rate region $\cS(\cG)$ is the set of all vectors $\blambda\in\bbRp^k$ for which there exists $\alpha_{Y}$, satisfying the following constraints: 
\begin{subequations}\label{eq:2}
\begin{align}\label{condition_demand_3}
&\sum_{Y\in\cY_i} \alpha_{Y}\geq\lambda_i, ~~~~~~~~~~~ \text{for all} ~~~ i\in [k], \\
&\sum_{\substack{{Y\in\cY} \\  \label{condition_capacity_3} {\hspace{0.05cm} j\in Y}}} \alpha_{Y} \leq n_j, ~~~~~~~~~~~ \text{for all} ~~~ j\in [\ell],\\
&\hspace{0.1cm}\alpha_{Y} \in \mathbb{R}_{\geq 0},~~~~~~~~~~~~~~~ \text{for all} ~~~ i\in [k],~Y\in\cY_i.\label{condition_pos_3}
\end{align}
\end{subequations}


\subsection{Problem Statement}

After these preparations, we can state the problems that we explore to address in this paper. Consider a practical scenario where we are asked to store $k$ files redundantly across some number of nodes in a coded distributed storage system. Also, we are given a bounded subset $\mathcal{R} \subset \mathbb{R}^k_{\geq 0}$ as a desired service rate region for this distributed storage system. Two natural questions arising in the design of this storage system are the following: 1) What is the minimum number $n(\mathcal{R})$ of storage nodes (or servers) required for serving all demand vectors $\boldsymbol{\lambda}$ inside the desired service rate region $\mathcal{R}$? 2) What is the most storage-efficient redundancy scheme with the service rate region covering the set $\mathcal{R}$ (i.e., how should the files be stored redundantly in $n(\mathcal{R})$ storage nodes)?

In other words, for each desired service rate region $\mathcal{R}\subset\mathbb{R}_{\ge 0}^k$, the goal is to characterize the minimum number of nodes $n(\mathcal{R})$ (or derive a lower bound on $n(\mathcal{R})$) such that there exists a generator matrix $\bG$ with $\mathcal{R}\subseteq\mathcal{S}(\bG)$ (or alternatively, a multiset of points $\mathcal{G}$ in $\operatorname{PG}(k-1,q)$ with $\mathcal{R}\subseteq\mathcal{S}(\mathcal{G})$). Thus, deriving lower bounds and constructive upper bounds for $n(\mathcal{R})$ is of great significance in the context of designing distributed storage systems, which we aim to address in this paper.

Recall that $\cS(\bG)\neq \cS(\bG')$ (or $\cS(\cG)\neq \cS(\cG')$) even if $\bG,\bG'$ (or $\cG,\cG'$) generate the same linear code. Thus, to be more precise, instead of designing an efficient code, we have to speak of the construction of storage-efficient generator matrices or multisets of points.

\section{Main Results}

In this section, first we investigate a few structural properties and formulate the problem of determining $n(\cR)$. Then, using a geometric approach, we derive multiple lower bounds on $n(\cR)$ and finally we show that for $k=2$ the derived lower bounds are tight by proposing an storage-efficient redundancy scheme.

\subsection{Structural Properties of the Service Rate Region}

Here, before we present integer linear programming (ILP) formulations for the determination of $n(\cR)$ we first study a few structural properties.   

\begin{lemma}\label{prop:1}
  For $\mathcal{R} \subset \mathbb{R}^k_{\geq 0}$, we have $n(\cR)=n(\conv(\cR))$.
\end{lemma}

\begin{definition}
For a set $S\subseteq \mathbb{R}_{\ge 0}^k$, the lower set $S\!\!\downarrow$ is defined as 
$S\!\!\downarrow:=\left\{\bx\in\mathbb{R}_{\ge 0}^k \mid \exists \by\in S: \bx\le \by\right\}$. 
\end{definition}

Consider the bounded subset $S=\conv(\{(0,0),(1,2),(2,1)\})\subset\mathbb{R}_{\ge 0}^2$ which is a triangle with area $\tfrac{3}{2}$. The corresponding lower set $S\!\!\downarrow=\conv(\{(0,0),(0,2),(2,0),(1,2),(2,1)\})$ is a pentagon with area $\tfrac{7}{2}$.

\begin{lemma}\label{prop:2} 
For a subset $\mathcal{R} \subset \mathbb{R}^k_{\geq 0}$, we have $n(\cR)=n(\cR\!\!\downarrow)$.
\end{lemma}

Taken the above two observations into account, we can parameterize a large class of reasonable subsets $\mathcal{R}\subset\mathbb{R}_{\ge 0}^k$ through a function $T\colon 2^{[k]}\to\mathbb{N}$ that maps the subsets of $[k]$ to integers in $\mathbb{N}$, where $T(\emptyset)=0$.

\begin{definition}\label{def:regionfunction}
Let  $T\colon 2^{[k]}\to\mathbb{N}$ with $T(\emptyset)=0$. We define
  $$
    \cR(T):=\left\{\blambda\in\mathbb{R}_{\ge 0}^k \mid \sum_{i\in S} \lambda_i\le T(S), \,\ \forall \emptyset\neq S\subseteq [k] \right\}
  $$
\end{definition}

By construction $\cR(T)$ is a convex polytope and $\cR(T)\!\!\downarrow=\cR(T)$, i.e., $\cR(T)$ is its own lower set. (For more details, see e.g.,~\cite{ServiceGeometric:KazemiKS20}.) It should be noted that in some cases, the values of the function ${T\colon 2^{[k]}\to\mathbb{N}}$ can be modified without changing $\cR(T)$. 

\begin{lemma}\label{lemma_characterization_fixpoint_algo_transform}
For each function $T\colon 2^{[k]}\to\mathbb{N}$ with $T(\emptyset)=0$, there exists a monotone and subadditive function $T'\colon 2^{[k]}\to\mathbb{N}$ with $T'(\emptyset)=0$, such that $R(T)=R(T')$.\footnote{$T\colon 2^{[k]}\to\mathbb{N}$ is monotone iff $T(U)\le T(V)$ for all $\emptyset\subseteq U\subseteq V\subseteq [k]$, and is subadditive iff $T(U\cup V) \leq T(U)+T(V)$.}
\end{lemma} 

\begin{definition}\label{def:genset}
 For a subset $\cR\subset\mathbb{R}_{\ge 0}^k$ with property $\cR\!\!\downarrow=\cR$, a finite set $S\subset\mathbb{R}_{\ge 0}^k$ is a generating set of $\cR$ if $\conv(S)\!\!\downarrow=\cR$. Moreover, we call $S$ minimal if no proper subset of $S$ is a generating set of $\cR$. 
\end{definition}

In what follows, without explicitly mentioning, we only consider the minimal generating sets for each ${\cR\subset\mathbb{R}_{\ge 0}^k}$. As an example, consider the function ${T\colon 2^{[2]}\to\mathbb{N}}$ given by ${T(\emptyset)=0}$, $T(\{1\})=T(\{2\})=2$, and ${T(\{1,2\})=3}$. Here, a generating set of $\cR(T)$ is given by the set ${\{(1,2),(2,1)\}}$. 

We remark that the generating set of $\cR(T)$ is always unique, since $\cR(T)$ is a polytope that can be written as $\cR(T)=\conv(V)$, where $V$ is the set of vertices of the polytope, which is a unique minimal set. The unique generating set of $\cR(T)$ is obtained from $V$ by removing all vectors $\bv\in V$ such that there exists a vector $\bv'\in V$ with $\bv\le \bv'$. 

Before we study bounds for $n(\cR(T))$, we present an ILP formulation for the determination of $n(\cR)$.

\begin{proposition}\label{prop_ILP_n_exact}
For a desired service rate region $\mathcal{R} \subset \mathbb{R}^k_{\geq 0}$, assume ${\cR\!\!\downarrow=\cR}$. Let ${\left\{\blambda^{(1)},\dots,\blambda^{(m)}\right\}}$ be the generating set of $\cR$. Then, $n(\cR)$ coincides with the optimal target value of
\begin{eqnarray}\label{eq:ILP-n(R)}
&\min &\sum_{j\in [\ell]} n_j\\ 
&\text{s.t.} &\sum_{Y\in\mathcal{Y}^j} \alpha^{i}_Y \ge \lambda^{(i)}_j ~~~~~~~ \forall i\in [m],\, j\in [k]\nonumber \\
&&\sum_{\substack{{Y\in\cY} \\   {\hspace{0.05cm} j\in Y}}} \alpha^{i}_Y\le n_j, ~~~~~~~~~ \forall j\in [\ell], \forall i\in [m]\nonumber\\ 
&&\hspace{0.1cm}\alpha^{i}_Y\in\bbRp,~~~~~~~~~~~~ \forall i\in [m], \forall Y\in\cY\nonumber\\
&&\hspace{0.1cm}n_j\in \bbN,~~~~~~~~~~~~~~~~ \forall j\in [\ell]\nonumber
\end{eqnarray}where $\lambda^{(i)}_j$ is the $j$th element of the $\blambda^{(i)}$ and $\alpha^{i}_Y$ is the portion of requests coming from $\blambda^{(i)}$ assigned to the recovery set $Y$.
\end{proposition}

The ILP formulation~\eqref{eq:ILP-n(R)} in the Proposition~\ref{prop_ILP_n_exact}, underlies a massive combinatorial explosion. To be more precise, when the number of files $k$ increases, the number of recovery sets $\#\cY$ grows doubly exponential, that is, $\#\cY$ gets quite large even for moderate values of $k$. In order to obtain a lower bound on $n(\cR)$, one simple way is to consider the ceiling of the optimal target value for the LP relaxation of the ILP~\eqref{eq:ILP-n(R)}. However, this approach again suffers from the same drawback and runs into a similar problem since to list all the constraints of the LP relaxation of the ILP~\eqref{eq:ILP-n(R)}, one needs to explicitly know all possible recovery sets which becomes increasingly complex when the number of files $k$ increases. Thus, introducing a technique which is not depending on the enumeration of recovery sets is of great significance. Towards this goal, we introduce a novel geometric approach.

\subsection{Using Geometric Approach to derive Bounds on $n(\mathcal{R})$}

Here, we present three lower bounds for $n(\cR(T))$ that are obtained using a geometric technique. 

\begin{lemma}\label{lemma_hyperplane_constraint}
Let $\bG$ be the generator matrix of an $[n,k]_q$ code and $\cG$ be the corresponding multiset of points of cardinality $n$ described by point multiplicities $n_j$ for $j \in [\ell]$. If ${\left\{\blambda^{(1)},\dots,\blambda^{(m)}\right\}}$ be the generating set of $\cR$, then we have
\begin{equation}\label{ie_hyperplane_constraint}
  \sum_{j\,:\, \bv_j\in \PG(k-1,2) \setminus \cH} n_j \ge
  \max \left\{ \sum_{s\in \cE(\cH)} \lambda^{(i)}_s \mid i\in [m]\right\},   
\end{equation}  
where $\cH$ is a hyperplane of $\PG(k-1,2)$ and 
  $$
    \cE(\cH)=\left\{h\in[k] \mid \be_h\notin \left\langle\left\{\bv\mid \bv\in\cH\right\}\right\rangle\right\}
  $$
is the set of indices $h$ such that the hyperplane $\cH$ does not contain the unit vector $\be_h$, i.e., $\be_h$ lies in 
  $\PG(k-1,2)\setminus\cH$ . 
\end{lemma}

\begin{Corollary}\label{cor_ILP_n_lb}
If $\left\{\blambda^{(1)},\dots,\blambda^{(m)}\right\}$ is the generating set of $\cR$, then $n(\cR)$ is lower bounded by the optimal target value of the following ILP formulation:
  \begin{eqnarray}\label{ILP_geometric}
  &\min &\sum_{j\in [\ell]} n_j\\
  &\text{s.t.} &(5) ~~\text{holds}~~\forall ~\text{hyperplane}~\cH~\text{of}~\PG(k-1,2) \nonumber\\
  &&n_j\in \bbN ~~~~~~\forall j\in [\ell].\nonumber 
 \end{eqnarray}
\end{Corollary}

It should be noted that the ILP of Corollary~\ref{cor_ILP_n_lb} contains $2^k-1$ constraints and (integer) variables. Thus, with respect to the LP relaxation of the ILP formulation~\eqref{eq:ILP-n(R)} in Proposition~\ref{prop_ILP_n_exact}, we have obtained a smaller formulation for the determination of a lower bound on $n(\cR)$.  

\begin{definition}
Let $P=\left\{\bx\in \bbR^k \mid \mathbf{A}\bx\le \mathbf{b}, \bx\ge 0\right\}$ be a polytope in $\bbR^k$ with description $(\mathbf{A},\mathbf{b})$. We say that constraint ${\mathbf{a}^{(i)}\bx\le b_i}$ is redundant, where $\mathbf{a}^{(i)}$ is the $i$th row of $\mathbf{A}$, if ${P=\left\{\bx\in \bbR^k \mid \mathbf{A'}\bx\le \mathbf{b'}, \bx\ge 0\right\}}$ where $\mathbf{A'}$ and $\mathbf{b'}$ obtained from $\mathbf{A}$ and $\mathbf{b}$ by removing the $i$th row, respectively. We say that a constraint $\mathbf{a}^{(i)}\bx\le b_i$ is strictly redundant if there does not exist $\bar{\bx}\in P$ with $\mathbf{a}^{(i)}\bar{\bx}=b_i$. 
\end{definition}

For example, consider ${T\colon 2^{[2]}\to\bbN}$ defined as ${T(\emptyset)=0}$, $T(\{1\})=T(\{2\})=T(\{1,2\})=1$. Consider the polytope $P=\left\{\blambda\in \bbR^2 \mid {\sum_{i\in U}\lambda_i\le T(U)}, {\emptyset\neq U\subseteq\{1,2\}}, \blambda\ge 0\right\}$. The inequalities $\lambda_1\le T(\{1\})$, $\lambda_2\le T(\{2\})$ are redundant, while the inequality $\lambda_1+\lambda_2\le T(\{1,2\})$ is not redundant since e.g.\ the vector $(1,1)$ is not contained in the polytope. Here, none of the inequalities are strictly redundant since the vectors $(1,0)$, $(0,1)$ are contained in the polytope. 

\begin{theorem}\label{theorem_general_lower_bound}
Given the function $T\colon 2^{[k]}\to\bbN$ for some $k \in \bbN$, we have 
  $$
    n(\cR(T))\ge \left\lceil\frac{\sum_{\emptyset\neq U\subseteq[k]} T(U)}{2^{k-1}}\right\rceil,
  $$
where none of the constraints $\sum_{i\in U}\lambda_i \le T(U)$ are strictly redundant in ${\bbRp^k}$. 
\end{theorem}

As we will show shortly the lower bound of Theorem~\ref{theorem_general_lower_bound} is indeed tight if $k=2$ and $T\colon 2^{[k]}\to\bbN$ is monotone and subadditive. However, this bound is not tight in general for $K \geq 3$. The following example shows that for $K = 3$ this bound is not tight, while none of the constraints are strictly redundant.

\begin{example}\label{ex_k_3}
For $k=3$, consider the desired service rate region $\cR=\cR(T)$ for $T\colon 2^{[3]}\to\bbN$ defined as $T(\emptyset)=0$ and $T(S)=\#S+1$ for $\emptyset\neq S\subseteq[3]$, that is, $\cR$ is as follows
  $$  
    \cR=\left\{\blambda\in\bbRp^3\,:\, 
    \lambda_1\le 2,\lambda_2\le 2,\lambda_3\le 2,\lambda_1+\lambda_2\le 3,\lambda_1+\lambda_3\le 3,\lambda_2+\lambda_3\le 3, \lambda_1+\lambda_2+\lambda_3\le 4\right\}.
  $$
A generating set $\left\{\blambda^{(1)},\blambda^{(2)},\blambda^{(3)}\right\}$ of $\cR$ of cardinality $m=3$ is given by $\blambda^{(1)}=(2,1,1)$, $\blambda^{(2)}=(1,2,1)$, and $\blambda^{(3)}=(1,1,2)$. The possible columns of a generator matrix $\bG$, i.e., the non-zero vectors in $\bbF_2^3$ are
  $$
    \bv_1=(0,0,1), \bv_2=(0,1,0), \bv_3=(0,1,1), \bv_4=(1,0,0), \bv_5=(1,0,1), \bv_6=(1,1,0), \bv_7=(1,1,1).
  $$ 
In order to write down the inequalities from Lemma~\ref{lemma_hyperplane_constraint} we describe a hyperplane $\cH$ as a set of vectors $\left(x_1,x_2,x_3\right)\in\bbF_2^3\setminus\{\mathbf{0}\}$ satisfying a certain constraint $\sum_{i=1}^3 c_ix_i$, where $\left(c_1,c_2,c_3\right)\in\bbF_2^3\setminus\{\mathbf{0}\}$:
  \begin{eqnarray}
  \label{ie_ex_k_3_1}\cH_1:\,\, x_1=0 &\!\!\Rightarrow\!\!& \be_1 \notin \cH_1 ~\Rightarrow~ n_4+n_5+n_6+n_7 \ge 2= \max(\lambda^{(1)}_1,\lambda^{(2)}_1,\lambda^{(3)}_1)\\
  \label{ie_ex_k_3_2}\cH_2:\,\, x_2=0 &\!\!\Rightarrow\!\!& \be_2 \notin \cH_2 ~\Rightarrow~ n_2+n_3+n_6+n_7 \ge 2= \max(\lambda^{(1)}_2,\lambda^{(2)}_2,\lambda^{(3)}_2)\\
  \label{ie_ex_k_3_3}\cH_3:\,\, x_3=0 &\!\!\Rightarrow\!\!& \be_3 \notin \cH_3 ~\Rightarrow~ n_1+n_3+n_5+n_7 \ge 2= \max(\lambda^{(1)}_3,\lambda^{(2)}_3,\lambda^{(3)}_3)\\
  \label{ie_ex_k_3_4}\cH_4:\,\, x_1+x_2=0 &\!\!\Rightarrow\!\!& \be_1,\be_2 \notin \cH_4 ~\Rightarrow~ n_2+n_3+n_4+n_5 \ge 3= \max\big(\sum_{j=1,2}\lambda^{(i)}_j : i \in [3]\big)\\
  \label{ie_ex_k_3_5}\cH_5:\,\, x_1+x_3=0 &\!\!\Rightarrow\!\!& \be_1,\be_3 \notin \cH_5 ~\Rightarrow~ n_1+n_3+n_4+n_6 \ge 3= \max\big(\sum_{j=1,3}\lambda^{(i)}_j : i \in [3]\big)\\
  \label{ie_ex_k_3_6}\cH_6:\,\, x_2+x_3=0 &\!\!\Rightarrow\!\!& \be_2,\be_3 \notin \cH_6 ~\Rightarrow~ n_1+n_2+n_5+n_6 \ge 3= \max\big(\sum_{j=2,3}\lambda^{(i)}_j : i \in [3]\big)\\
  \label{ie_ex_k_3_7}\cH_7\!: x_1\!+\!x_2\!+\!x_3\!=\!0 &\!\!\!\Rightarrow\!\!\!& \be_1,\be_2,\be_3 \notin \cH_7 ~\Rightarrow~ n_1\!+\!n_2\!+\!n_4\!+\!n_7 \ge 4= \max\big(\!\sum_{j=[3]}\!\lambda^{(i)}_j : i \in [3]\big)
  \end{eqnarray}   
Summing up inequalities \eqref{ie_ex_k_3_1}-\eqref{ie_ex_k_3_7} and dividing by four gives $n\ge\left\lceil\tfrac{19}{4}\right\rceil=5$. Indeed, the LP relaxation of the ILP~\eqref{ILP_geometric} from Corollary~\ref{cor_ILP_n_lb} has an optimal solution $n_1=n_2=n_4=\tfrac{5}{4}$, $n_3=n_5=n_6=n_7=\tfrac{1}{4}$ with target value $\tfrac{19}{4}$. Next, we show that  $n\ge 6$ for the optimal target value of ILP~\eqref{ILP_geometric}. Assume that there exists an integral solution with $n=5$. Summing the inequalities over all hyperplanes $\cH_i$ containing $\bv_1=\be_3$, i.e., \eqref{ie_ex_k_3_1}, \eqref{ie_ex_k_3_2}, and \eqref{ie_ex_k_3_4}, and dividing by two gives $\sum_{j\in[\ell]\setminus\{1\}} n_j\ge 3.5$, so that $n_1\le 1$. By symmetry, we also conclude $n_2,n_4\le 1$. Summing the inequalities over all hyperplanes $\cH_i$ not containing $\bv_1=\be_3$, i.e., \eqref{ie_ex_k_3_3}, \eqref{ie_ex_k_3_5}, \eqref{ie_ex_k_3_6}, and \eqref{ie_ex_k_3_7}, and dividing by two gives $2n_1+\sum_{j\in[\ell]\setminus\{1\}} n_j\ge 6$, so that $n_1\ge 1$. Thus, $n_1=1$ and, by symmetry, also $n_2=n_4=1$. Summing inequalities \eqref{ie_ex_k_3_4}-\eqref{ie_ex_k_3_6}, plugging in the known values, and dividing by two gives $n_3+n_5+n_6\ge 1.5$, so that $n_7\le 0.5$, i.e., $n_7=0$. However, this contradicts Inequality~\eqref{ie_ex_k_3_7}.
  
An integral solution for $n=6$ can indeed be attained by $n_1=n_2=n_4=2$, $n_3=n_5=n_6=n_7=0$. It can be easily checked that the corresponding generator matrix $\bG$ as given below satisfies $\cS(\bG)\supseteq\cR$.  
  $$
    \bG=\begin{pmatrix}
    1 & 1 & 0 & 0 & 0 & 0 \\
    0 & 0 & 1 & 1 & 0 & 0 \\
    0 & 0 & 0 & 0 & 1 & 1
   \end{pmatrix}
  $$
\end{example}

\begin{Corollary}\label{coro_simplex}
For some $k \in \bbN$ and $X\in \bbN$, given the function $T\colon 2^{[k]}\to\bbN$ defined as $T(\emptyset)=0$, $T(U)=X$ for all subsets $\emptyset\neq U\subseteq [k]$, we have
  $$
    n(\cR(T))\ge \left\lceil\frac{X\cdot\left(2^k-1\right)}{2^{k-1}} \right\rceil.
  $$
Moreover, if $X=t\cdot 2^{k-1}$ for some integer $t$, then the lower bound is tight.
\end{Corollary}

Next, two more general lower bounds for ${n(\cR(T))}$, similar to that of Theorem~\ref{theorem_general_lower_bound}, are provided that are obtained in the search of finding a tighter lower bound for ${k \geq 3}$.

\begin{theorem} \label{theorem_general_lower_bound2}
For some integer $k\ge 2$, let $T\colon 2^{[k]}\to\bbN$  be a function such that none of the constraints $\sum_{i\in U}\lambda_i \le T(U)$ are strictly redundant in $\bbRp^k$. Then, for each ${i\in[k]}$, ${n(\cR(T))\ge \left\lceil\frac{\alpha_i+\beta_i}{2}\right\rceil}$ holds, where
  $$
    \alpha_i=\left\lceil\frac{\sum_{\emptyset\neq U\subseteq[k]\setminus\{i\}} T(U)}{2^{k-2}}\right\rceil,~ \beta_i=\left\lceil\frac{\sum_{\{i\}\subseteq U\subseteq[k]} T(U)}{2^{k-2}}\right\rceil.
  $$ 
\end{theorem}

\begin{theorem}\label{theorem_general_lower_bound3}
For some integer $k\ge 2$, let $T\colon 2^{[k]}\to\bbN$  be a function such that none of the constraints $\sum_{i\in U}\lambda_i \le T(U)$ are strictly redundant in $\bbRp^k$. Then, for each $j\in[\ell]$ we have 
  $$
    n(\cR(T))\ge \left\lceil\frac{\sum_{\emptyset\neq U\subseteq[k] \,:\, \#(U\cap J)\equiv 0 \pmod 2} T(U)}{2^{k-2}}\right\rceil,
  $$  
where $J\subseteq [k]$ such that $\bv_j=\sum_{h\in J} \be _h$.
\end{theorem}

\begin{example}\label{ex_negative_1}
For some $x \in \bbN$, let $T\colon 2^{[3]}\to\bbN$ be defined via $T(\{1\})=T(\{2\})=T(\{3\})=T(\{1,2\})=T(\{1,3\})=x$ and $T(\{2,3\})=T(\{1,2,3\})=2x$. Based on Theorem~\ref{theorem_general_lower_bound}, we have $n(\cR(T))\ge\left\lceil\tfrac{9x}{4}\right\rceil$, and according to Theorem~\ref{theorem_general_lower_bound3}, considering $j=3$ we have ${n(\cR(T))\ge\left\lceil\tfrac{5x}{2}\right\rceil}$. Thus, for $x\geq3$, the lower bound obtained from Proposition~\ref{theorem_general_lower_bound3} is tighter than the one obtained from Theorem~\ref{theorem_general_lower_bound}. 
\end{example}

\subsection{Storage-Efficient Schemes for $k=2$}

Let w.l.o.g. (based on Lemma~\ref{lemma_characterization_fixpoint_algo_transform}) the function ${T\colon 2^{[2]}\to\bbN}$ be monotone, subadditive, and satisfy $T(\emptyset)=0$. Note that for ${k=1}$ each ${T\colon 2^{\{1\}}\to\mathbb{N}}$ is monotone and subadditive, while for ${k=2}$ the conditions can be summarized to
\begin{equation}\label{ie_monotone_subadditive_k_2}\nonumber
  \max\{T(\{1\}),T(\{2\})\} \le T(\{1,2\}) \le  T(\{1\})+T(\{2\}).  
\end{equation}  
The following Lemma describes the generating set of $\cR(T)$ for ${T\colon 2^{[2]}\to\bbN}$.
\begin{lemma}\label{lemma_characterization_generating_set_k_2}
If ${T\colon 2^{[2]}\to\bbN}$ is monotone, subadditive, and satisfies ${T(\emptyset)=0}$, the generating set of ${\cR(T)}$ is given by
\begin{align*}
S=\Big\{ \Big(T(\{1\}),T(\{1,2\})-T(\{1\})\Big),\Big(T(\{1,2\})-T(\{2\}),T(\{2\})\Big) \Big\}.
\end{align*}
\end{lemma}

We remark that the cardinality of generating set of ${\cR(T)}$ in Lemma~\ref{lemma_characterization_generating_set_k_2} is $2$ or $1$, where the latter happens iff $T(\{1,2\})=T(\{1\})+T(\{2\})$.

\begin{lemma}\label{lemma_gen_size1_k2}
Let $\{\blambda\}$ be the generating set of $\cR$ and $\bn=(n_1,\cdots,n_{\ell})$ be an integral solution of the ILP of Corollary~\ref{cor_ILP_n_lb}. If $\blambda\in \bbRp^2$ and $\cG$ is the multiset corresponding to the $\bn$, then $\blambda\in\cS(\cG)$, i.e., there exists a feasible choice of $\alpha_Y$ satisfying (\ref{condition_demand_3})-(\ref{condition_pos_3}).      
\end{lemma}

\begin{definition}\label{def_simplex_code}
For a set ${\emptyset\neq S\subset\bbN}$, we denote by $\simplex(S)$ the set of non-zero vectors in $\left\langle\left\{\be_i\mid i\in S\right\}\right\rangle$ over $\bbF_2$. 
\end{definition}

The following Proposition states \cite[Theorem1]{ServiceGeometric:KazemiKS20}, adapted to the form useful
for our setup.
 
\begin{proposition}\label{lemma_service_rate_region_simplex}
For each ${\emptyset\neq S\subseteq[k]}$, $\#\simplex(S)=2^{s}-1$ and $\cS(\simplex(S))=\cR(T)$, where $s=\#S$ and $T\colon 2^{[k]}\to\bbN$ is given by $T(U)=2^{s-1}$ for all $U\subseteq[k]$ satisfying $U\cap S\neq \emptyset$ and $T(U)=0$ otherwise (for all $U\subseteq[k]$ with $U\cap S = \emptyset$).
\end{proposition}

\begin{theorem}\label{thm_two_server}
For the desired service rate region $\cR$ given by
  \begin{align*}
    \cR= \Big\{ \blambda \in \bbRp^2 \,:\, \lambda_1\le X, \lambda_2\le Y,\lambda_1+\lambda_2\le\Sigma \Big\}, 
  \end{align*}
  where $X,Y,\Sigma$ are non-negative integers with $\max\{X,Y\}\le\Sigma\le X+Y$, we have $n(\cR)=\left\lceil \frac{X+Y+\Sigma}{2}\right\rceil$.
\end{theorem}

\appendix[Proofs of Lemmas, Theorems and Corollaries]

\begin{proof}[Proof of Lemma~\ref{prop:1}]
  It suffices to observe that the service rate region $\cS(\bG)$ of every generator matrix $\bG\in\mathbb{F}_2^{k\times n}$ is convex. For more details, see~\cite{ServiceGeometric:KazemiKS20}.
\end{proof}

\begin{proof}[Proof of Lemma~\ref{prop:2}]
  It suffices to observe that the service rate region $\cS(\bG)$ of every generator matrix $\bG\in\mathbb{F}_2^{k\times n}$ is its own lower set, i.e., $\cS(\bG)=\cS(\bG)\!\!\downarrow$. 
\end{proof}

\begin{proof}[Proof of Lemma~\ref{lemma_characterization_fixpoint_algo_transform}]
Consider the function $T\colon 2^{[k]}\to\mathbb{N}$ with $T(\emptyset)=0$. Let $T'$ be obtained as the result of the following algorithm applied to $T$.

\begin{algorithmic}
  \FOR {each $S\subseteq\{1,\dots,k\}$}
    \STATE {$T'(S)\leftarrow T(S)$}
  \ENDFOR
  \STATE{ $changed\leftarrow \text{\texttt{true}}$}
  \WHILE {$changed=true$}
    \STATE{ $changed\leftarrow \text{\texttt{false}}$}
    \FOR {each $S\subseteq \{1,\dots,k\}$}
      \FOR {each $\emptyset\neq U\subsetneq S$}
        \IF {$T'(S)>T'(U)+T'(S\backslash U)$}
          \STATE {$T'(S)\leftarrow T'(U)+T'(S\backslash U)$}
          \STATE{ $changed\leftarrow \text{\texttt{true}}$}
        \ENDIF
      \ENDFOR
      \FOR {each $S\subsetneq V\subseteq \{1,\dots,k\}$}
        \IF {$T'(S)>T'(V)$}
          \STATE {$T'(S)\leftarrow T'(V)$}
          \STATE{ $changed\leftarrow \text{\texttt{true}}$}
        \ENDIF  
      \ENDFOR
    \ENDFOR
  \ENDWHILE    
\end{algorithmic}
We remark that based on the above algorithm, the function $T'$ is subadditive, i.e., we have $T'(U\cup V) \leq T'(U)+T'(V)$, and monotone, that is, we have $T'(U)\le T'(V)$ for all $\emptyset\subseteq U\subseteq V\subseteq [k]$. Now, we need to prove that $\cR(T)=\cR(T')$. 

After the first initializing loop we obviously have $\cR(T)=\cR(T')$. Now, let us consider a single step in which $T'(S)$ is replaced by either $T'(U)+T'(S\backslash U)$ or $T'(V)$. Inductively we know that each $\blambda\in\cR(T')$ satisfies $\sum_{i\in S'} \lambda_i\le T'(S')$ for all $S'\subseteq [k]$. Since this especially holds for ${S'=U}$, ${S'=S\backslash U}$, and ${S'=V}$ we also have 
  $$
    \sum_{i\in S} \lambda_i\le T'(U)+T'(S\backslash U)
  $$
and
  $$
    \sum_{i\in S} \lambda_i\overset{\blambda\ge 0}{\le} \sum_{i\in V} \lambda_i\le T'(V).
  $$
So, after each replacement we still have $\cR(T)=\cR(T')$. In order to show that the algorithm terminates let 
  $$
    \varepsilon=\min\{T(U)-T(V)\mid \emptyset\subseteq U,V\subseteq [k],T(U)-T(V)\}.
  $$ 
By induction over the number of replacements we can easily show that at each time after the initialization loop , we have
  $$\varepsilon\le \min\{T'(U)-T'(V)\mid \emptyset\subseteq U,V\subseteq [k], T'(U)-T'(V)\}
  $$ 
Thus, every replacement reduces the value of $\sum_{S\subseteq[k]} T'(S)$ by at least $\varepsilon$, so that the algorithm terminates after at least  $(\sum_{S\subseteq[k]} T(S))/\varepsilon+1$ iterations of the while loop. 

Note that since in the last iteration of the while loop non of the if-conditions were true, if we apply the algorithm again on $T'$ and obtain $T''$, then $T'=T''$. 
\end{proof}

\begin{proof}[Proof of Proposition~\ref{prop_ILP_n_exact}]
The multiset of points $\cG$ is uniquely characterized by the integer multiplicities $n_j$, $j\in[\ell]$. Thus, it is easy to verify that the stated ILP formulation minimizes the code length $n=\sum_{j\in[\ell]} n_j$ and ensures that $\blambda^{(i)}\in\cS(\cG)$ by using the characterization \eqref{condition_demand_3}--\eqref{condition_pos_3} for each $i\in[m]$. 
\end{proof}

\begin{proof}[Proof of Lemma~\ref{lemma_hyperplane_constraint}]
Let $i\in [m]$ be an arbitrary index. From the ILP of Proposition~\ref{prop_ILP_n_exact}, we conclude that 
  \begin{equation}
    \label{ie_hpc_1}
    \sum_{Y\in\mathcal{Y}_s} \alpha^{i}_Y \ge \lambda^{(i)}_s
  \end{equation}
Since ${\alpha^i_Y\ge 0}$, for each $s\in\cE(\cH)$ we have
  \begin{equation}
    \label{ie_hpc_2}
    n_j\ge \sum_{Y\in \mathcal{Y}\,:\, j\in Y} \alpha^{i}_Y{\ge} \sum_{s\in\cE(\cH)}\,\,\sum_{Y\in \mathcal{Y}_{s}\,:\, j\in Y} \alpha^{i}_Y 
  \end{equation}

Let $J=\{j\in[\ell] \mid \bv_j\in \PG(k-1,2) \setminus \cH\}$. Thus, we have
\begin{align*}
   \sum_{j \in J} n_j\,&\ge\, \sum_{j \in J} \sum_{s\in\cE(\cH)}\,\,\sum_{Y\in \mathcal{Y}_{s}\,:\, j\in Y} \alpha^{i}_Y 
     \\
    &= \sum_{s\in\cE(\cH)} \sum_{j \in J} \sum_{Y\in \mathcal{Y}_{s}\,:\, j\in Y} \alpha^{i}_Y.   
\end{align*}

The unit vectors $\be_s$ with index $s \in \cE(\cH)$ are not contained in the hyperplane $\mathcal{H}$. Thus, for each $Y\in \mathcal{Y}_{s}$ with $s\in\cE(\cH)$ there exists certainly an index $j\in[\ell]$ with $j\in Y$ such that $\bv_j\in \PG(k-1,2) \setminus \cH$. Thus, from Inequality~(\ref{ie_hpc_1}) we conclude
 $$
  \sum_{j\in [\ell]\,:\, \bv_j\in \PG(k-1,2) \setminus \cH} n_j\,\ge\, \sum_{s\in\cE(\cH)} \sum_{Y\in \mathcal{Y}_{s}} \alpha^{i}_Y\ge \sum_{s\in\cE(\cH)} \lambda^{(i)}_s  
 $$    

\end{proof}

\begin{proof}[Proof of Theorem~\ref{theorem_general_lower_bound}]
We observe that each hyperplane $\cH$ in $\PG(k-1,2)$ can be uniquely characterized by a subset $\emptyset \neq S(\cH)\subseteq[k]$ such that $S(\cH)=\left\{i\in[k]\mid \be_i\notin\cH\right\}$. Let w.l.o.g. $\left\{\blambda^{(1)},\dots,\blambda^{(m)}\right\}$ be the generating set of $\cR(T)$. Due to the fact that none of the constraints $\sum_{i\in U}\lambda_i \le T(U)$ is strictly redundant, we have
$$
  \max \left\{ \sum_{s\in S(\cH)} \lambda^{(i)}_s \mid i\in [m]\right\}= T(S(\cH))
$$
  
Thus, for each hyperplane $\cH$ of $\PG(k-1,2)$, by applying Lemma~\ref{lemma_hyperplane_constraint} and replacing $T(S(\cH))$ in the right hand side of inequality~\eqref{ie_hyperplane_constraint}, we get 
$$
  \sum_{j\in [\ell]\,:\, \bv_j\in \PG(k-1,2) \setminus \cH} n_j \ge T(S(\cH)),
$$  
where ${S(\cH)=\left\{i\in[k]\,:\, \be_i\notin\cH\right\}}$. Since there are $2^k-1$ hyperplanes in $\PG(k-1,2)$, we have $2^k-1$ such inequalities, each of which can be uniquely characterized by subset $S(\cH)$. For each ${j\in[\ell]}$, one can easily verify that ${\bv_j\notin\cH}$ for exactly $2^{k-1}$ hyperplanes $\cH$. This means that each $n_j$ for ${j\in[\ell]}$ appears in the left side of $2^{k-1}$ inequalities. Thus, summing all of the $2^k-1$ inequalities, dividing by $2^{k-1}$ and replacing variable $S(\cH)$ with variable $U$, yields 
$$
  n=\sum_{j\in[\ell]} n_j\ge \frac{\sum_{\emptyset\neq U\subseteq[k]} T(U)}{2^{k-1}}.
$$
Finally, since $n$ has to be an integer, we have
$$
  n\ge \lceil\frac{\sum_{\emptyset\neq U\subseteq[k]} T(U)}{2^{k-1}}\rceil.
$$
\end{proof}

\begin{proof}[Proof of Corollary~\ref{coro_simplex}]
One can easily check that none of the constraints ${\sum_{i\in U}\lambda_i \le T(U)}$ are strictly redundant in ${\bbRp^k}$. Thus, Theorem~\ref{theorem_general_lower_bound} can be applied. Thus,
$$
  n\ge \Big\lceil\frac{\sum_{\emptyset\neq U\subseteq[k]} T(U)}{2^{k-1}}\Big\rceil=\Big\lceil\frac{X\cdot\left(2^k-1\right)}{2^{k-1}}\Big\rceil.
$$
The generating set of $\cR(T)$ is given by $\left\{X\cdot \be_i\mid i\in[k]\right\}$. Thus, for $X=t\cdot 2^{k-1}$, a $t$-fold $k$-dimensional binary simplex code achieves the desired service rate region. For more details, see~\cite{ServiceGeometric:KazemiKS20}.
\end{proof}

\begin{proof}[Proof of Theorem~\ref{theorem_general_lower_bound2}]
Based on the same reasoning used in the proof of Theorem~\ref{theorem_general_lower_bound}, it can be shown that for each hyperplane $\cH$ of $\PG(k-1,2)$, we have
\begin{equation}\label{ie_hyp_bound_T}
  \sum_{j\in [\ell]\,:\, \bv_j\in \PG(k-1,2) \setminus \cH} n_j \ge T(S(\cH)),
\end{equation}  
where $S(\cH)=\left\{i\in[k]\,:\, \be_i\notin\cH\right\}$. Let $i\in[k]$ be arbitrary but fix and ${\bar{i}=2^{k-i}}$ such that ${\bv_{\bar{i}}=\be_i}$. Then, we proceed by summing Inequality~\eqref{ie_hyp_bound_T} for all subsets $\emptyset \neq S(\cH)\subseteq[k]\setminus\{i\}$ and replacing $S(\cH)$ everywhere with $U$, that gives
\begin{equation}\label{alpha_bound}
    2^{k-2}\cdot \sum_{j\in[\ell]\setminus\{\bar{i}\}} n_j\ge \sum_{\emptyset\neq U\subseteq[k]\setminus\{i\}} T(U).
\end{equation}  
Now, summing Inequality~\eqref{ie_hyp_bound_T} for all ${\{i\} \subseteq S(\cH)\subseteq[k]}$ and replacing $S(\cH)$ everywhere with $U$ gives
\begin{equation}\label{beta_bound}
    2^{k-1} \cdot n_{\bar{i}} +2^{k-2}\cdot \sum_{j\in[\ell]\setminus\{\bar{i}\}} n_j\ge \sum_{\{i\}\subseteq U\subseteq[k]} T(U).
\end{equation} 
Since the $n_j$s are integers, from \eqref{alpha_bound} and \eqref{beta_bound}, we get $\sum_{j\in[\ell]\setminus\{\bar{i}\}} n_j\ge \alpha_i$ and $2n_{\bar{i}}+\sum_{j\in[\ell]\setminus\{\bar{i}\}} n_j\ge\beta_i$, respectively, where
 $$
    \alpha_i=\left\lceil\frac{\sum_{\emptyset\neq U\subseteq[k]\setminus\{i\}} T(U)}{2^{k-2}}\right\rceil, ~~
    \beta_i=\left\lceil\frac{\sum_{\{i\}\subseteq U\subseteq[k]} T(U)}{2^{k-2}}\right\rceil.
 $$
 
Dividing the sum of these two inequalities by $2$ gives
  $$
    n=n_{\bar{i}}+\sum_{j\in[\ell]\setminus\{\bar{i}\}} n_j\ge \frac{\alpha_i+\beta_i}{2},
  $$   
Since $n$ has to be an integer, we have $n \ge \left\lceil\frac{\alpha_i+\beta_i}{2}\right\rceil
$.
\end{proof}

\begin{proof}[Proof of Theorem~\ref{theorem_general_lower_bound3}]
Similar to the proof of Theorems \ref{theorem_general_lower_bound} and \ref{theorem_general_lower_bound2}, it can be shown that for each hyperplane $\cH$ in $\PG(k-1,2)$, we have
\begin{equation}\label{ie_hyp_bound_T3}
    \sum_{i\in [\ell]\,:\, \bv_i\in \PG(k-1,2) \setminus \cH} n_i \ge T(S(\cH)),
\end{equation}  
where $S(\cH)=\left\{i\in[k]\,:\, \be_i\notin\cH\right\}$. For each index $i \in [\ell]$, it can be easily confirmed that $\bv_i\notin\cH$ for $2^{k-1}$ hyperplanes $\cH$ and ${\bv_i\in\cH}$ for $2^{k-1}-1$ hyperplanes $\cH$ of $\PG(k-1,2)$. Let $j\in[\ell]$ be arbitrary but fix. Our aim is to sum Inequality~\eqref{ie_hyp_bound_T3} over all ${2^{k-1}-1}$ hyperplanes $\cH$ that contain $\bv_j$. 

We proceed by proving a claim that  $\bv_j=\sum_{h\in J}\be_h\in\cH$ iff $\#(U\cap J)\equiv 0 \pmod 2$, where $U=S(\cH)$. We consider two cases: (i) If $\#U=1$, then w.l.o.g. assume $U=\{x\}$ for some $x\in[k]$. A basis of $\cH$ is given by $\left\{\be_y\mid y\in[k]\setminus\{x\}\right\}$. Thus, $\bv_j\in \cH$ iff $x\notin J$, i.e., $\#(U\cap J)\equiv 0 \pmod 2$. (ii) If $\#U\ge 2$, then for some arbitrary element $x\in U$, a basis of $\cH$ is given by the set
  $$
    \left\{\be_y\mid y\in [k]\backslash U\right\}\cup \left\{\be_x+\be_z\mid z\in U\setminus \{x\}\right\}
  $$
In this case, it is easy to see that ${\bv_j\in \cH}$ iff $\#(U\cap J)=2p$ for some $p \in \bbZ_{\geq 0}$, i.e., $\#(U\cap J)\equiv 0 \pmod 2$. Thus, the claim is proved. Now, by summing Inequality~\eqref{ie_hyp_bound_T3} over all hyperplanes $\cH$ that contain $\bv_j$, we obtain 
\begin{align*}
    2^{k-2} \cdot \sum_{i\in[\ell]\setminus \{j\}} n_i&=\sum_{\text{hyperplane }\cH\,:\,\bv_j\in\cH}~\sum_{i \in [\ell]\,:\, \bv_i\notin\cH} n_i \\
    &\ge 
    \sum_{\emptyset\neq U\subseteq[k] \,:\, \#(U\cap J)\equiv 0 \pmod 2} T(U).   
\end{align*}
Since $n \geq \sum_{i\in[\ell]\setminus \{j\}} n_i$ and $n$ is an integer, we have 
  $$
    n \ge \left\lceil\frac{\sum_{\emptyset\neq U\subseteq[k] \,:\, \#(U\cap J)\equiv 0 \pmod 2} T(U)}{2^{k-2}}\right\rceil.
  $$ 
\end{proof}

\begin{proof}[Proof of Lemma~\ref{lemma_characterization_generating_set_k_2}]
According to the Definition~\ref{def:regionfunction}, $\cR(T)$ is the set of all vectors ${\blambda\in\mathbb{R}_{\ge 0}^2}$ that satisfy $\lambda_1\le T(\{1\}), \lambda_2\le T(\{2\})$ and $\lambda_1+\lambda_2 \le T(\{1,2\})$. Based on Definition~\ref{def:genset}, $S\subset\mathbb{R}_{\ge 0}^2$ is a generating set of $\cR(T)$ if $\conv(S)\!\!\downarrow=\cR(T)$. 

The proof consists of two parts. First, we need to show that $\conv(S)\!\!\downarrow \subseteq\cR(T)$. For this purpose, we check that each $\blambda\in S$ satisfies the constraints $\lambda_1\le T(\{1\})$, $\lambda_2\le T(\{2\})$, and $\lambda_1+\lambda_2\le T(\{1,2\})$, i.e., $\blambda\in \cR(T)$. Thus, $S\subseteq \cR(T)$. Due to convex property of $\cR(T)$ and since $\cR(T)\!\!\downarrow=\cR(T)$, it can be easily concluded that $\conv(S)\!\!\downarrow \subseteq\cR(T)$. 

Now, for the other direction, we need to show that $\cR(T) \subseteq \conv(S)\!\!\downarrow$. Let $\blambda\in\bbRp^2$ satisfy the constraints $\lambda_1\le T(\{1\})$, $\lambda_2\le T(\{2\})$, and $\lambda_1+\lambda_2\le T(\{1,2\})$. W.l.o.g.\ we assume that at least one of these three inequalities is satisfied with equality, since we could increase $\blambda$ otherwise. If $\lambda_1+\lambda_2=T(\{1,2\})$, then it can be readily concluded that $\blambda\in\conv(S)$ since $\lambda_1\le T(\{1\})$ and $\lambda_2\le T(\{2\})$. Thus, let us now consider the case where $\lambda_1=T(\{1\})$. If $\lambda_2<T(\{2\})$ and $\lambda_1+\lambda_2<T(\{1,2\})$ then we could increase $\blambda$, so that we can assume $\lambda_2<T(\{2\})$ and conclude $\lambda_1+\lambda_2=T(\{1,2\})$ due to the subadditivity of $T$. The case $\lambda_2=T(\{2\})$ can be treated analogously. Thus, ${\cR(T) \subseteq \conv(S) \subseteq \conv(S)\!\!\downarrow}$. 
\end{proof}

\begin{proof}[Proof of Lemma~\ref{lemma_gen_size1_k2}]
The constraints of ILP in Corollary~\ref{cor_ILP_n_lb} are
  \begin{eqnarray*}
    n_2+n_3 &\ge &\lambda_1,\\ 
    n_1+n_3 &\ge &\lambda_2,\\
    n_1+n_2 &\ge &\lambda_1+\lambda_2
  \end{eqnarray*}
and the recovery sets are given by
  \begin{eqnarray*} 
    \cY_1 &=& \Big\{\{2\},\{1,3\}\Big\},\\
    \cY_2 &=& \Big\{\{1\},\{2,3\}\Big\}.
  \end{eqnarray*}
Now, set the parameters as follows
  \begin{eqnarray*}
    \alpha_{\{2\}} &=& \min\left\{n_2,\lambda_1\right\},\\ 
    \alpha_{\{1\}} &=& \min\left\{n_1,\lambda_2\right\},\\
    \alpha_{\{1,3\}} &=& \max\left\{0,\lambda_1-n_2\right\},\\
    \alpha_{\{2,3\}} &=& \max\left\{0,\lambda_2-n_1\right\}
  \end{eqnarray*}          
It should be noted that since $n_1+n_2\ge \lambda_1+\lambda_2$, we cannot have ${n_2<\lambda_1}$ and ${n_1<\lambda_2}$. Thus, it is easy to verify that
  \begin{eqnarray*}
    \alpha_{\{2\}}+\alpha_{\{1,3\}} &=& \lambda_1,\\
    \alpha_{\{1\}}+\alpha_{\{2,3\}} &=& \lambda_2,\\
    \alpha_{\{1\}}+ \alpha_{\{1,3\}}  &\le& n_1,\\  
    \alpha_{\{2\}}+ \alpha_{\{2,3\}}  &\le& n_2,\text{ and}\\
    \alpha_{\{1,3\}}+ \alpha_{\{2,3\}}  &\le& n_3.
  \end{eqnarray*}
Only the latter inequality needs a short case analysis. Let us assume $n_2\ge \lambda_1$ and $n_1\ge\lambda_2$, then $\alpha_{\{1,3\}}+ \alpha_{\{2,3\}}=0\le n_3$. If $n_2<\lambda_1$ and $n_1\ge \lambda_2$, then $\alpha_{\{2,3\}}=0$, $\alpha_{\{1,3\}}=\lambda_1-n_2$, and $\alpha_{\{1,3\}}+\alpha_{\{2,3\}}=\lambda_1-n_2$ which is at most $n_3$ due to $n_2+n_3 \ge \lambda_1$. The other case, that is, $n_2\ge\lambda_1$ and $n_1<\lambda_2$ follows analogously.   
\end{proof}

\begin{proof}[Proof of Theorem~\ref{thm_two_server}]
The proof consists of a converse (lower bound) and an achievability (upper bound).

\ul{Converse}: The desired service rate region $\cR$ is given by:
 \begin{align*}
  \cR=\left\{(\lambda_1,\lambda_2)\in\mathbb{R}_{\ge 0}^2\,:\, 
  \lambda_1\le X,\lambda_2\le Y,\lambda_1+\lambda_2\le \Sigma\right\}.   
  \end{align*}
  
It is easy to see that $\cR=\cR(T)$ for $T\colon 2^{[2]}\to\bbN$ defined as $T(\emptyset)=0$, $T(\{1\})=X$, $T(\{2\})=Y$, and $T(\{1,2\})=\Sigma$. Since  $\max\{X,Y\}\le\Sigma\le X+Y$ holds, so the condition  $\max\{T(\{1\}),T(\{2\})\} \le T(\{1,2\}) \le  T(\{1\})+T(\{2\})$ is satisfied which means that the function $T\colon 2^{[2]}\to\bbN$ is monotone and subadditive. Thus, we can apply Lemma~\ref{lemma_characterization_generating_set_k_2} to obtain the generating set of ${\cR(T)}$ which is given by 
$$S=\Big\{\lambda^{(1)}=(X,\Sigma-X), \lambda^{(2)}=(\Sigma-Y,Y)\Big\}$$
The inequalities~\eqref{ie_hyperplane_constraint} from Lemma~\ref{lemma_hyperplane_constraint} read
  \begin{eqnarray*}
    n_1+n_3 &\ge& X=\max\{X,\Sigma-Y\},\\
    n_2+n_3 &\ge& Y=\max\{Y,\Sigma-X\},\\
    n_1+n_2 &\ge& \Sigma=\max\{\Sigma,\Sigma\},
  \end{eqnarray*} 
so that summing up and dividing by two gives
  $$
    n=n_1+n_2+n_3\ge\frac{X+Y+\Sigma}{2}.
  $$ 
Since $n$ is an integer, we obtain $n(\cR)\ge \left\lceil \frac{X+Y+\Sigma}{2}\right\rceil$.

We could also use Theorem~\ref{theorem_general_lower_bound} for proving the converse. Since $\max\{X,Y\}\le\Sigma\le X+Y$, it can be simply confirmed that none of the constraints $\lambda_1\le X$, $\lambda_2\le Y$, $\lambda_1+\lambda_2\le \Sigma$ are strictly redundant in ${\bbRp^2}$. Thus, by applying Theorem~\ref{theorem_general_lower_bound}, we directly get the stated lower bound.
  
\ul{Achievability}: First, for the ease of notation, let us define $\tfrac{1}{2}\cdot\simplex(\{i,j\})\triangleq\left\{\be_i,\be_j\right\}$ for two different positive integers $i$ and $j$. Note that the cardinality of $\tfrac{L}{2}\cdot\simplex(\{i,j\})$ for some ${L \in \bbZ_{\geq 0}}$, is computed as ${\left\lceil\tfrac{L}{2}\cdot\#\simplex(\{i,j\})\right\rceil=\left\lceil\tfrac{3L}{2}\right\rceil}$ and the service rate region of $\left\{\be_i,\be_j\right\}$ contains the service rate region of the $\simplex(\{i,j\})$ scaled by a factor of $\frac{1}{2}$, i.e.,
\begin{align*}
\cS\!\left( \left\{\be_i,\be_j\right\}\right)=\left\{\blambda\in\bbRp^k\mid \lambda_i\le 1,\lambda_j\le 1\right\}\supseteq 
  \left\{\blambda\in\bbRp^k\mid \lambda_i\le 1,\lambda_j\le 1,\lambda_i+\lambda_j\le 1\right\}.
\end{align*}

For the upper bound on $n(\cR)$, i.e., the constructive part, we need to select the multiplicities of $\be_1$, $\be_2$ and $\be_1+\be_2$ in $\cG$, a multiset of points in $\PG(2-1,2)$, such that $\cS(\cG)\supseteq\cR(T)$. Let $\cG=\cup_{i\in[3]}\cG^{(i)}$ where\vspace{0.2cm}

\begin{itemize}
    \item $\cG^{(1)}$ consists of $\Sigma-Y$ copies of $\simplex(\{1\})$\vspace{0.2cm}
    \item $\cG^{(2)}$ consists of $\Sigma-X$ copies of $\simplex(\{2\})$\vspace{0.2cm}
    \item ${\cG^{(3)}}$ consists of ${\tfrac{L}{2}}$ copies of $\simplex(\{1,2\})$\vspace{0.2cm}
\end{itemize}
where $L=X+Y-\Sigma$. Thus, the cardinality of the multiset $\cG$ is given by
  $$
    \left(\Sigma-Y\right) +\left(\Sigma-X\right)+\left\lceil\frac{3(X+Y-\Sigma)}{2}\right\rceil=\left\lceil\frac{X+Y+\Sigma}{2}\right\rceil.
  $$     
By construction, $\cR(\cG)\supseteq\cR(T)$ for ${T(\emptyset)=0}$, ${T(\{1\})=X}$, ${T(\{2\})=Y}$, and $T(\{1,2\})=\Sigma$, since
\begin{align*}
    \cS(\cG^{(1)})\supseteq\left\{\blambda\in\bbRp^2\mid \lambda_1\le \Sigma-Y,\lambda_2= 0\right\},
\end{align*}
\begin{align*}
    \cS(\cG^{(2)})\supseteq\left\{\blambda\in\bbRp^2\mid \lambda_1= 0,\lambda_2\le \Sigma-X\right\},
\end{align*}
\begin{align*}
    \cS(\cG^{(3)})\supseteq 
  \left\{\blambda\in\bbRp^2\mid \lambda_1\leq L,\lambda_2 \leq L, \lambda_1+\lambda_2\le L\right\}.
\end{align*}
Thus, it can be easily confirmed that 
\begin{align*}
    \cS(\cG)\supseteq 
  \left\{\blambda\in\bbRp^2\mid \lambda_1\leq X,\lambda_2 \leq Y, \lambda_1+\lambda_2\le \Sigma\right\}=\cR.
\end{align*}
That is, the proposed storage scheme $\cG$ obviously can satisfy the demands in $\cR$.
\end{proof}

\bibliographystyle{IEEEtran}
\bibliography{service_rate}

\end{document}